\documentclass[journal,10pt]{IEEEtran}

\pdfoutput=1

\usepackage{amsmath}
\usepackage{cite}
\usepackage{xcolor} 
\usepackage{algorithm}
\usepackage{algorithmicx}
\usepackage[noend]{algpseudocode} 
\usepackage{amsfonts}
\usepackage{graphicx}
\usepackage{epsfig}
\usepackage{amsthm}
\usepackage{hyperref}
\usepackage{subfigure}
\usepackage{extarrows}
\usepackage{lettrine}

\theoremstyle{definition}
\newtheorem{assumption}{Assumption}
\theoremstyle{definition}
\newtheorem{lemma}{Lemma}
\theoremstyle{plain}
\newtheorem{theorem}{Theorem}
\begin{document}

\title{\LARGE Communication-Efficient Model Aggregation with Layer Divergence Feedback in Federated Learning}

\author{Liwei Wang, Jun Li, Wen Chen, Qingqing Wu, and Ming Ding
\thanks{Liwei Wang, Wen Chen and Qingqing Wu are with the Broadband Access
Network Laboratory, Shanghai Jiao Tong University, Minhang 200240, China
(e-mail: wanglw2000; wenchen; qingqingwu@sjtu.edu.cn). Jun Li is with School of Electrical and Optical Engineering, Nanjing University of Science and Technology, Nanjing 210094, China
Ministry of Education (e-mail: jun.li@njust.edu.cn).
Ming Ding is with Data61, CSIRO, Sydney, NSW 2015, Australia (e-mail:
ming.ding@data61.csiro.au).}

}
\maketitle

\begin{abstract}
Federated Learning (FL) facilitates collaborative machine learning by training models on local datasets, and subsequently aggregating these local models at a central server. However, the frequent exchange of model parameters between clients and the central server can result in significant communication overhead during the FL training process. To solve this problem, this paper proposes a novel FL framework, the Model Aggregation with Layer Divergence Feedback mechanism (FedLDF). Specifically, we calculate model divergence between the local model and the global model from the previous round. Then through model layer divergence feedback, the distinct layers of each client are uploaded and the amount of data transferred is reduced effectively. Moreover, the convergence bound reveals that the access ratio of clients has a positive correlation with model performance. Simulation results show that our algorithm uploads local models with reduced communication overhead while upholding a superior global model performance.
\end{abstract}

\IEEEpeerreviewmaketitle

\section{Introduction}
\lettrine{F}{ederated} Learning (FL) is a distributed framework, which can obtain a global model by independently training local models on each client and aggregating their parameters on a central server\cite{McMahan2023}.
During the FL training process, the frequent exchange of model parameters can cause huge communication overhead, resulting in the client experience decline due to the greater latency and energy consumption. This challenge stems from the inherent limitations of clients, which are often equipped with constrained hardware resources and network bandwidth\cite{han2023analysis}, thereby hindering the deployment of large and deep neural networks. Consequently, considering the limited communication resources between the clients and the server, achieving more efficient communication has become a pivotal concern in the development and implementation of FL.

Prior studies have employed model compression techniques to mitigate the communication overhead associated with FL. Model compression seeks to enhance the speed of model reasoning and conserve storage space by diminishing the volume of neural networks. For instance, authors in \cite{nam2022} re-parameterized weight parameters of layers using low-rank weights to overcome the burdens on frequent model uploads and downloads. 
A new quantization method was explored in \cite{han2024energy} to reduce energy consumption and communication overhead, with allocating the quantization level adaptively. Additionally, a dataset-aware dynamic pruning approach \cite{yu2023} was proposed to accelerate the inference on edge devices and overcome the  communication redundancy issues. The work in \cite{Liu2023globalcom} dynamically determines the pruning ratio of each layer of the client to reduce the consumption of communication resources. 
Furthermore, the authors of \cite{zawad2023hdfl} introduced random dropout into client selection for each round of aggregation to reduce communication consumption.

However, most of the above sparse or pruning methods apply traditional machine learning methods directly to the FL framework. Authors in \cite{cheng2022does} pointed out that dropout does not seem to be an effective way for the distributed architecture of FL based on their experimentation. Furthermore, \cite{kairouz2021advances} suggested that the independent operation of parameters across different layers within FL might impact the overall model performance. Additionally, the application of fine density pruning and quantization to individual layers in these studies, introduces an extra computational overhead.


In FL, models from each client make distinct contributions, necessitating the need to tailor the model for individual clients. A scheme was proposed in \cite{tao2023communication} to evaluate the channel-wise parameter importance dynamically via a fast Taylor series evaluation. \cite{rehman2023} weighted each layer of the client and aggregates them but without accounting for the communication overhead. To reduce communication overhead while maintaining high global model performance in FL, we proposed a novel model aggregation method (FedLDF), which takes into account the layer differences among clients to reduce the communication overhead effectively, and our primary contributions can be summarized as follows:
\begin{itemize}
\item[$\bullet$] We introduce a new FL aggregation mechanism, FedLDF. By incorporating model divergence feedback into the traditional FL framework, clients do not contribute equivalently, involving different layers uploading.
\item[$\bullet$] We employ a novel and simpler convergence analysis method to conduct a fundamental analysis on an expression for the expected convergence rate of FedLDF. The analysis reveals that the accessing ratio of clients significantly influences the convergence speed.
\item[$\bullet$]We conduct experiments to assess the communication efficiency. The experimental results reveal that the new aggregation mechanism achieves significant reductions of $80\%$ in communication overhead across different layering situations while maintaining high model performance.
\end{itemize}

\section{Framework and convergence analysis}

\subsection{FedLDF Algorithm}
 It is assumed that there are $N$ distributed clients, each with its own local dataset $\{D_1,D_2,\cdots,D_N\}$ and local models $\{\Theta_1^t,\Theta_2^t,\cdots,\Theta_N^t\}$ at $t$-th communication round. FL executes the procedures of clients' local training and server's periodic model aggregation. At each communication round $t$, the server selects a set of participators with $K$ clients as $\mathcal{C}_t$. With FedAvg\cite{McMahan2023}, the parameter aggregation process is:
\begin{equation}  
	\bar{\Theta}^t = \sum_{k\in \mathcal{C}_t} \frac{\vert D_k \vert}{\sum_{m\in \mathcal{C}_t} \vert D_m \vert} \Theta_k^t,
\end{equation}
where $\Bar{\Theta}^t$ is the global model after the $t$-th global aggregation, $D_k$ is the datasets of the client $k$, $\vert \cdot \vert$ represents the size of the datasets, and $\Theta_{k,t}$ is the local model of client $k$ after the $t$-th local training. The client $k$ computing the gradient to update the local model parameters with its local datasets is as follows:
\begin{equation}
    \Theta_k^t = \Theta_k^{t-1} - \eta \nabla F_k\left(\Theta_k^{t-1} \right),
\end{equation}
where $F_k$ is the local loss function computed by the client $k$, $\nabla(\cdot)$ expresses the gradient and $\eta$ is the learning rate.

In traditional machine learning, as it would result in the loss of the model and degradation in performance, we can not afford to abandon any layer of the neural network, which differs from the FL framework. Hence, we designed FedLDF algorithm whose feasibility is derived from the special distributed architecture of FL. The primary objective of FedLDF is to \textit{upload local models with lower communication overhead while maintaining a high global model performance in the FL process}. As illustrated in Fig. \ref{flsp}, unlike the general FL framework, the server receives different network layers for each client due to the model divergence feedback, which describes the difference between the local model and the global model of the last round.



\begin{figure}[!t]
\centering
\includegraphics[width=0.4\textwidth]{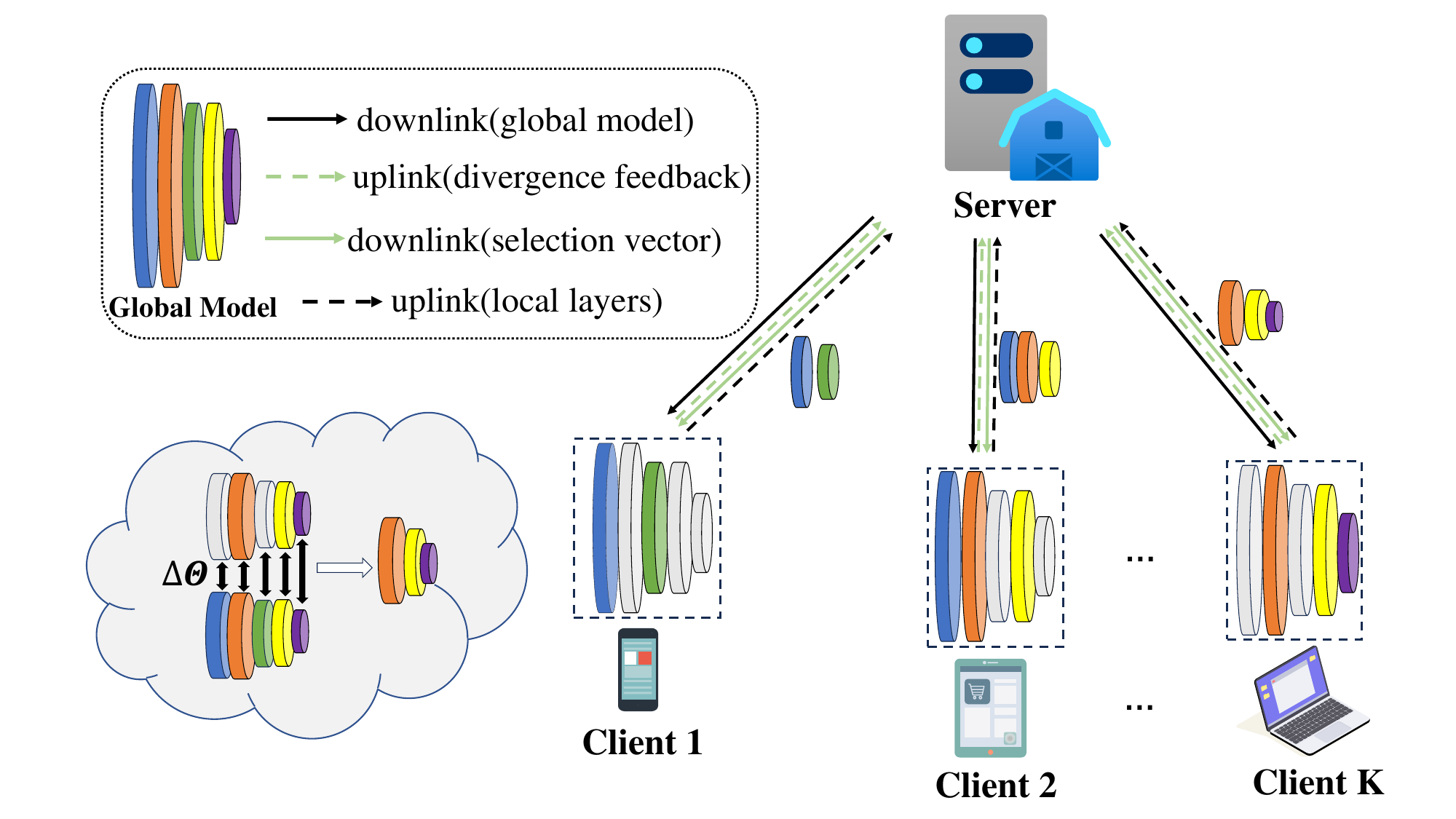}
\caption{FedLDF framework.}
\label{flsp}
\end{figure}

It is worth noting that our algorithm is applied before model aggregation process. Step1: At each communication round, the server broadcasts its global model $\hat{\boldsymbol{\Theta}}^t$ to all clients, and clients train their models with their local data. Step2: FedLDF selects a set of $K$ clients as $\mathcal{C}_t$ randomly. After one time local training, each client obtains its local model $\Theta_k^t$. We assume each client has its local model $\Theta$ with $L$ layers, which may consist of Convolutional Neural Network (CNN) layers and Fully Connected (FC) layers. The local model is represented as $\Theta_k^t=[\Theta_{k,1}, \Theta_{k,2}, \cdots,\Theta_{k,l}, \cdots, \Theta_{k,L}]$, where $\Theta_{k,l}$ represents the $l$-th layer's parameters of client $k$. Step3: clients in $\mathcal{C}_t$ calculate the divergence for each layer between their trained model and the global model of the previous communication round:
\begin{equation}
    \Delta\Theta_{k,l}^{t+1} =\left\Vert \Theta_{k,l}^{t+1} - \hat{\Theta}_{l}^{t}\right\Vert = \sqrt{\sum_{i \in l} \left(\Theta_{k,li}^{t+1} - \hat{\Theta}_{li}^{t} \right)^2},
    \label{eqdivergence}
\end{equation}
where $\hat{\boldsymbol{\Theta}}_{l}^{t}$ represents the $t$-th rounds that is the last global model before this local training process and $\Vert \cdot \Vert$ is $L2$-norm. $\Theta_{k,li}$ is the $i$-th neuron parameter of the $l$-th layer from client $k$. All gap scalars $\Delta\Theta_{k,1}^{t+1},\Delta\Theta_{k,2}^{t+1},\cdots,\Delta\Theta_{k,L}^{t+1}$ are stacked into a divergence vector $\Delta\Theta_{k}^{t+1}$.

As shown in Fig. \ref{flsp}, FedLDF involves two additional transmission steps compared to the traditional FL framework, represented by the green lines. The green dotted line indicates that the participating clients upload $\Delta \boldsymbol{\Theta}_k^{t+1}$ to the server side, i.e., model layer divergence feedback step.It is easy to understand that the greater the divergence, the greater the contribution to the model updating\cite{rehman2023}. We take top-$n$ for each row of the divergence matrix to get the client selection vector $\boldsymbol{s}_k^l$, which determines the layers from the trained client to upload to the server and complete the aggregation process. The element $s_k^l \in \{0,1\}$ indicates whether layer $l$ of client $k$ is allocated in the global aggregation:
\begin{equation}
     s_k^l = \begin{cases}
                1,& \text{the index of $s_k^l$ is in the top-$n$, }\\
                0,& \text{otherwise.}
            \end{cases}
    \label{eqselection}
\end{equation}
Herein, after server generates the selection vector, FedLDF modifies the traditional FL aggregation mechanism according to the new rule:
\begin{equation}
    \hat{\Theta}_l^t = \sum_{k\in \mathcal{C}_t} \frac{s_k^l  \vert D_k\vert \cdot \Theta_{k,l}^t}{\sum_{m\in \mathcal{C}_t}\vert D_m\vert s_m^l} ,
    \label{eqaggregation1}
\end{equation}
\begin{equation}
    \hat{\Theta}^t = \left[\hat{\Theta}_1^t,\hat{\Theta}_2^t,\cdots,\hat{\Theta}_l^t,\cdots,\hat{\Theta}_L^t\right].
    \label{eqaggregation2}
\end{equation}
Considering this, the updated global model $\hat{\boldsymbol{\Theta}}^t$ is formed by the layer model $\hat{\boldsymbol{\Theta}}_l^t$ aggregated by the chosen clients. The entire algorithm flow is illustrated in Algorithm \ref{alg:FLSP}, and each client contributes its different layers of the model in Fig. \ref{layer}.
\begin{algorithm}[htb]
    \renewcommand{\algorithmicrequire}{\textbf{Input:}}
    \renewcommand{\algorithmicensure}{\textbf{Output:}}
    \caption{FedLDF Algorithm}
    \label{alg:FLSP}
    \begin{algorithmic}[1] 
        \Require  dataset$\{D_1,D_2,\cdots,D_N\}$, learning rate $\eta$, total communication rounds $T$; 
	    \Ensure  global model $\hat{\Theta}^t$; 
         \State \textbf{Initialize}:the initial global model $\hat{\Theta}^0$ and system parameters;
         \Procedure{ServerExecute}{}
            \For {communication round $t=1,2,\cdots,T$}
                \State Broadcast global model $\hat{\Theta}^t$ to all clients;
                \State $\mathcal{C}_t \leftarrow random(K, max(C*N, 1))$;
                \For{each client $k$ in $\mathcal{C}_t$ \textbf{in parallel}} 
                    
                    
                    \State $\hat{\Theta}_k^{t+1} \leftarrow ClientUpdate(k, \hat{\Theta}^t)$;
                    \State Calculate divergence $\Delta \Theta$ by \eqref{eqdivergence};
                    \State Generate a client selection
                    \Statex \qquad\qquad\qquad vector $\boldsymbol{s}_k^l$ for each layer as \eqref{eqselection};
                \EndFor
                \State Aggregate each layer wisely with $\boldsymbol{s}_k^l$ by \eqref{eqaggregation1}\eqref{eqaggregation2};

            \EndFor
        \EndProcedure

        \Procedure{ClientUpdate}{}
            \State Client $k$ receives $\hat{\Theta}^{t+1}$ from the server;
            \State Set $\Theta_k \leftarrow \hat{\Theta}^{t+1}$;
                        \State $\Theta_k \leftarrow \Theta_k - \eta \nabla F_k\left(\Theta_k\right)$;

            \State \textbf{return} $\Theta_k$
        \EndProcedure
    \end{algorithmic}

\end{algorithm}
\begin{figure}[!t]
\centering
\includegraphics[width=0.5\textwidth]{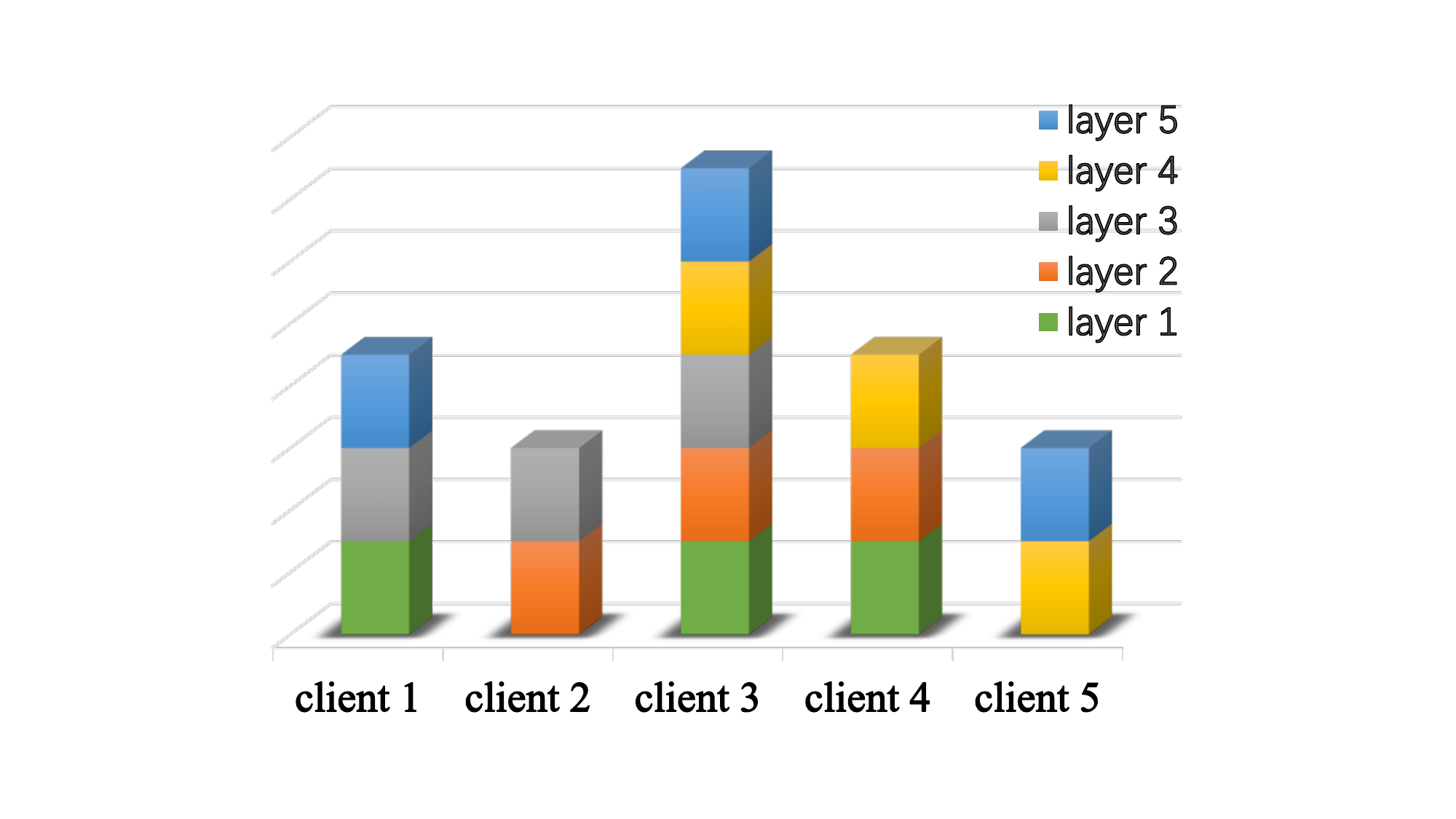}
\caption{Under the five-layer neural network of five clients, each client uploads the different local model layers in the aggregation$(n=3,K=5)$.}
\label{layer}
\end{figure}

We observe that the FedLDF algorithm considers the layer as the fundamental pruning unit, making it well-suited for various layer-dominated networks such as DNN, VGGnet, ResNet, and so on. Furthermore, FedLDF not only diminishes communication consumption but also conserves computing resources for edge devices, as clients are relieved from the burden of determining the model type in FL.

\subsection{Convergence Analysis}
In this section, we delve into the convergence behaviour of our FedLDF learning algorithm. To simplify the formula, we consider that the dataset size is equal for each client. We propose \textit{a novel and simple convergence analysis method}. It is well-established that the FedAvg algorithm is convergent with the convergence rate being $\mathcal{O}(\frac{1}{T})$\cite{fedavgconvergence}. Hence we directly analyze the divergence between the global loss functions of FedLDF and FedAvg, i.e., $F(\hat \Theta)$ and $F(\bar \Theta)$, to obtain the convergence result of FedLDF. 

To facilitate analysis, we assume that the following assumptions are satisfied, which are widely used in previous works on the convergence analysis.
\begin{assumption}
    ($\beta$-smooth) We assume that $F(\Theta)$ is convex and $\beta$-smooth.
\end{assumption}
\begin{assumption}
    (Gradient Divergence) We assume that the relationship between the gradient of the client $k$'s $l$-th layer $\nabla F_{k,l}(\Theta)$ and the global gradient $\nabla F(\Theta)$ is as followed\cite{mingzhechen2021}, with $\xi_1,\xi_2 \geq 0$.
    \begin{equation}
        ||\nabla F_{k,l}(\Theta)||^2 \leq \xi_1+\xi_2||\nabla F(\Theta)||^2.
        \label{assump2}
    \end{equation}
\end{assumption}
\begin{assumption}
    (Bounded Gradient) We assume that the value of the global gradient has an upper bound, i.e.,
    \begin{equation}
        ||\nabla F(\Theta)||^2 \leq G^2.
        \label{assump3}
    \end{equation}
\end{assumption}

We suppose the global aggregation of FedLDF occurs after each client completes one round of local updating. Following a similar approach to \cite{ShiqiangWang2019}, we suppose the model parameters of FedAvg as assisted variables, which share the same parameter starting point before each local updating with our algorithm. The expected convergence rate of the algorithm can now be obtained by the following theorem.

\begin{theorem}
    Given the number of clients $K$ participated in local training, the number of clients $n$ that upload the $l$-th layer and the learning rate $\eta$, with the assumptions, when the condition $0 < \xi_2 < \frac{1}{2(1+\beta)\eta^2 L^2}$ is satisfied, then the convergence upper bound of Algorithm 1 is given by
    \begin{small}
    \begin{equation}
    F(\hat{\Theta}^{t+1})-F( \Bar{\Theta}^{t+1}) \leq A^t \left[F(\hat{\Theta}^0)-F( \Bar{\Theta}^0) \right]+ B\frac{1-A^t}{1-A},
    \end{equation}
    \end{small}where A = $2\xi_2\eta^2 L^2(1-\frac{n}{K})\left[1+\beta(1-\frac{n}{K}) \right]$ and B = $\frac{\xi_1}{\xi_2}A+(1-\frac{n}{K})\frac{G^2}{2}$.
    \label{theorem1}
\end{theorem}
\begin{proof} [Proof of Theorem \ref{theorem1}]
    We first derive \textbf{Lemma \ref{lemma1}}, which can play a key role in our subsequent derivation.
    \begin{lemma}
    \textit{With the same variables and conditions as Theorem \ref{theorem1}, the layer model gap between FedLDF and FedAvg at the (t+1)-th communication round is given by}
    \begin{small}
    \begin{align}
        || \hat{\Theta}_l^{t+1} &- \Bar{\Theta}_l^{t+1} ||^2 \notag\\
        \leq& 4\eta^2(1-\frac{n}{K})^2\left(\xi_1+\xi_2||\nabla F(\hat{\Theta}^t)||^2 \right).
    \end{align}
    \end{small}
    \label{lemma1}
\end{lemma}
\begin{proof}[Proof of Lemma \ref{lemma1}]
    \begin{small}
    \begin{align}
    & \hat{\Theta}_l^{t+1} - \Bar{\Theta}_l^{t+1} = \frac{\sum_k s_k^l \hat{\Theta}_{k,l}^{t+1}}{n}-(\hat{\Theta}_l^t-\eta \nabla F_{.,l}(\hat{\Theta}^t)). \notag
    \end{align}
    \end{small}Noting that these calculations are for a certain layer $l$ of the model, and $\sum_k$ means going over the sum from $1$ to $K$. For $\sum_k s_k^l=n$, and $\hat{\Theta}_{k,l}^{t+1}=\hat{\Theta}_l^t- \eta \nabla F_{k,l}(\hat{\Theta}^t)$, then we have
    \begin{small}
    \begin{align}
    &\quad \hat{\Theta}_l^{t+1} - \Bar{\Theta}_l^{t+1} \notag\\
    =&\frac{\eta \sum_k s_k^l}{n} \nabla F_{k,l}(\hat{\Theta}^t)-\eta \nabla F_{.,l}(\hat{\Theta}^t)  \notag\\
    =&\frac{\eta}{n}\left( \sum_k s_k^l \nabla F_{k,l}(\hat{\Theta}^t)-\frac{n\sum_k\nabla F_{k,l}(\hat{\Theta}^t)}{K} \right). \label{eqlemma1}
    \end{align}
    \end{small}We divide $K$ clients into two sets of clients who participate in the aggregation of a certain layer and the others not. If the $l$-th layer of client $k$ is selected into the current round of aggregation, i.e., $s_k^l=1$ then $k \in U_{in}^l$ and $s_k^l=0$ else $k\in U_{out}^l$, otherwise. Hence, \eqref{eqlemma1} can be rewritten as 
    \begin{small}
    \begin{align}
        &\left \Vert\hat{\Theta}_l^{t+1} - \Bar{\Theta}_l^{t+1} \right\Vert^2 \notag\\
        =& \frac{\eta^2}{n^2}\left \Vert\sum_{k\in U_{in}^l}(1-\frac{n}{K}) \nabla F_{k,l}(\hat{\Theta}^t) -\sum_{k\in U_{out}^l}\frac{n}{K}\nabla F_{k,l}(\hat{\Theta}^t)\right\Vert^2\notag\\
        \overset{(a)}{\leq}&\frac{\eta^2}{n^2}\left[\sum_{k\in U_{in}^l}(1-\frac{n}{K}) \left \Vert\nabla F_{k,l}(\hat{\Theta}^t)\right\Vert+\sum_{k\in U_{out}^l}\frac{n}{K} \left \Vert\nabla F_{k,l}(\hat{\Theta}^t)\right\Vert\right]^2\notag\\
        \leq& \frac{\eta^2}{n^2} \Bigg[ n(1-\frac{n}{K})\sqrt{\xi_1+\xi_2\left\Vert\nabla F(\hat{\Theta}^t)\right\Vert^2} \notag\\
        &\qquad +\frac{n}{K}(K-n)\sqrt{\xi_1+\xi_2\left\Vert\nabla F(\hat{\Theta}^t)\right\Vert^2} \Bigg]^2 \notag\\
        =&4\eta^2(1-\frac{n}{K})^2(\xi_1+\xi_2\left\Vert\nabla F(\hat{\Theta}^t)\right\Vert^2). \label{fangsuo}
    \end{align}
    \end{small}\end{proof}The triangle inequality is applied to inequality (a). This completes the proof of \textbf{Lemma \ref{lemma1}}. And then we start with the proof of \textbf{Theorem \ref{theorem1}}.As Assumption 1, we have
    \begin{small}
    \begin{align}
        F(\hat{\Theta}&^{t+1})-F( \Bar{\Theta}^{t+1}) \notag\\
        \leq& \nabla F^T(\Bar{\Theta}^{t+1})(\hat{\Theta}^{t+1}- \Bar{\Theta}^{t+1})+\frac{\beta}{2}\left\Vert\hat{\Theta}^{t+1} - \Bar{\Theta}^{t+1}\right\Vert^2 \notag\\
        \leq&\left\Vert\nabla F(\Bar{\Theta}^{t+1})\right\Vert\cdot \left\Vert\hat{\Theta}^{t+1}- \Bar{\Theta}^{t+1}\right\Vert+\frac{\beta}{2}\left\Vert\hat{\Theta}^{t+1} - \Bar{\Theta}^{t+1}\right\Vert^2. \label{loss_gap}
    \end{align}
    \end{small}Noticing that applying triangle inequality again, we have
    \begin{small}
    \begin{align}
        ||\hat{\Theta}^{t+1} - \Bar{\Theta}^{t+1}|| =& ||\sum_{l=1}^L(\hat{\Theta}_l^{t+1} - \Bar{\Theta}_l^{t+1})|| \notag \\
        \leq& L\cdot \max_l ||\hat{\Theta}_l^{t+1} - \Bar{\Theta}_l^{t+1}||. \notag
    \end{align}
    \end{small}And then we can substitute the above inequation into \eqref{loss_gap}, divide $1-\frac{n}{K}$ into $\sqrt{1-\frac{n}{K}}$ squared and apply \textbf{Lemma \ref{lemma1}}, continuing to deduce that
    \begin{small}
    \begin{align}
        &  F(\hat{\Theta}^{t+1})-F( \Bar{\Theta}^{t+1}) \leq L\left\Vert\nabla F(\Bar{\Theta}^{t+1})\right\Vert\cdot \max_l \left\Vert\hat{\Theta}_l^{t+1}- \Bar{\Theta}_l^{t+1}\right\Vert \notag\\
        &\quad+\frac{\beta L^2}{2}\max_l \left\Vert\hat{\Theta}_l^{t+1} - \Bar{\Theta}_l^{t+1}\right\Vert^2, \notag\\
        &\leq \sqrt{1-\frac{n}{K}}\left\Vert\nabla F(\Bar{\Theta}^{t+1})\right\Vert\cdot 2\eta L\sqrt{1-\frac{n}{K}}\sqrt{\xi_1+\xi_2\left\Vert\nabla F(\hat{\Theta}^t)\right\Vert^2}\notag\\
        &\quad+ 2\beta \eta^2 L^2(1-\frac{n}{K})^2(\xi_1+\xi_2\left\Vert\nabla F(\hat{\Theta}^t)\right\Vert^2) \notag\\
        &\overset{(b)}{\leq} \frac{1}{2}(1-\frac{n}{K})\left\Vert\nabla F(\Bar{\Theta}^{t+1})\right\Vert^2 \notag\\\
        &\quad+2\eta^2 L^2(1-\frac{n}{K})\left[1+\beta(1-\frac{n}{K})\right](\xi_1+\xi_2||\nabla F(\hat{\Theta}^t)||^2) . \label{eq12}
    \end{align}
    \end{small}The inequality equation (b) is achieved by the inequality of arithmetic and geometric mean. Because $F (\cdot)$ is $\beta$-smooth, subtracting the first two inequalities below, there are
    \begin{small}
    \begin{align}
        ||\nabla F(\hat{\Theta}^t)||^2 \leq 2\beta[F(\hat{\Theta}^t)-F(\Theta^\ast)], \notag\\
        ||\nabla F(\Bar{\Theta}^t)||^2 \leq 2\beta[F(\Bar{\Theta}^t)-F(\Theta^\ast)], \notag\\
         ||\nabla F(\hat{\Theta}^t)||^2\leq 2\beta[F(\hat{\Theta}^t)-F(\Bar{\Theta}^t)]. 
         \label{eq13}
    \end{align}
    \end{small}Substituting \eqref{assump2} and \eqref{eq13} into \eqref{eq12} and sorting out formula, then we have
    \begin{small}
    \begin{align}
          F(\hat{\Theta}^{t+1})-&F( \Bar{\Theta}^{t+1})\notag\\
        \leq&\frac{1}{2}(1-\frac{n}{K})G^2+2\xi_1\eta^2 L^2(1-\frac{n}{K})\left[1+\beta(1-\frac{n}{K})\right] \notag\\
        +&2 \xi_2\eta^2 L^2(1-\frac{n}{K}) \left[1+\beta(1-\frac{n}{K}) \right] \left[F(\hat{\Theta}^t)-F(\Bar{\Theta}^t)\right] \notag\\
        =&A\left[F(\hat{\Theta}^t)-F(\Bar{\Theta}^t) \right]+A\frac{\xi_1}{\xi_2}+(1-\frac{n}{K})\frac{G^2}{2},\label{eq14}
    \end{align}
    \end{small}where $A =2\xi_2\eta^2 L^2(1-\frac{n}{K})\left[1+\beta(1-\frac{n}{K}) \right]$.
    
    Applying \eqref{eq14} recursively, we have
    \begin{small} 
    \begin{align}
    F(\hat{\Theta}^{t+1})-&F( \Bar{\Theta}^{t+1}) \notag\\
    \leq& A^t\left[F(\hat{\Theta}^0)-F( \Bar{\Theta}^0) \right]+ B\frac{1-A^t}{1-A},
    \end{align}
    \end{small}
     where $B =\frac{\xi_1}{\xi_2}A+(1-\frac{n}{K})\frac{G^2}{2}$. And this completes the proof.
\end{proof}

In \textbf{Theorem \ref{theorem1}}, $\hat{\Theta}^{t+1}$ is the global FL model generated based on the local layers of selected clients($s_k^l=1$) and $\bar{\Theta}^{t+1}$ is the global model aggregated by the FedAvg both at the $(t+1)$-th communication round. According to \textbf{Theorem \ref{theorem1}}, to guarantee convergence of FedLDF algorithm, $A\leq1$ must be satisfied, i.e., $\xi_2<\frac{1}{2(1+\beta)\eta^2 L^2}=\min\{\frac{1}{2\eta^2 L^2(1-\frac{n}{K})\left[1+\beta(1-\frac{n}{K}) \right]}\}$. Since $\xi_2$ must satisfy Assumption 2, we must have $\xi_2 > 0$. From \textbf{Theorem \ref{theorem1}}, we can see there exists a gap between $F(\hat{\Theta}^{t+1})$ and $F(\bar{\Theta}^{t+1})$ which is caused by the selected clients numbers for each layer. When $t$ is large enough, the gap is $\frac{(1-\frac{n}{K})\frac{G^2}{2}+\frac{\xi_1}{\xi_2}}{1-A}-\frac{\xi_1}{\xi_2}$. We can see that as $n$ increases, $1-A$ increases then the gap narrows, indicating that the algorithm converges faster. When $n$ increases to $n=K$,  the gap vanishes (reaches 0), and FedLDF degenerates into FedAvg. The correctness of the theorem will be verified in Section \uppercase\expandafter{\romannumeral3}.

\section{Experiment Result}
\subsection{Experiment Settings}
To evaluate the effectiveness of our proposed approach, we conduct a series of experiments in an image classification FL task using the CIFAR-10 dataset. We set a total of 50 clients, and in every global round, 20 clients $(K=20)$ are randomly selected for aggregation and FedLDF will select 4 clients $(n=4)$ from the 20 clients.

We adopt a VGG-9 network model with 8 Conv layers and 1 linear layer (FC) as the global model, where batch normalization (BN) and max-pooling operations are conducted following each Conv layer. We categorize the distribution of the data into i.i.d. and Non-i.i.d. cases. In the i.i.d. case, 50,000 train images are allocated to all clients randomly, and each client obtains 1,000 samples of uniformly distributed classes. For the non-i.i.d. case, the Dirichlet distribution is applied, and we set $\alpha=1$, which means Non-i.i.d. data with different dataset sizes are allocated to each client. To emphasize the advanced nature of layer-smart by FedLDF, we add random algorithm, FedADP\cite{Liu2023globalcom} and HDFL\cite{zawad2023hdfl} to the baselines. The random algorithm selected clients for each layer randomly, but in HDFL\cite{zawad2023hdfl}, clients were randomly selected at each communication round of aggregation. FedADP\cite{Liu2023globalcom} proposed a FL framework for adaptive pruning with the neuron as the smallest pruning unit. For comparison purposes, we set the baselines pruning ratio to $0.2$ to get the same communication overhead with $T=1000$. It is noticing that the test error is used to describe the model performance.
\begin{figure}[!t]
\centering
\includegraphics[width=0.45\textwidth]{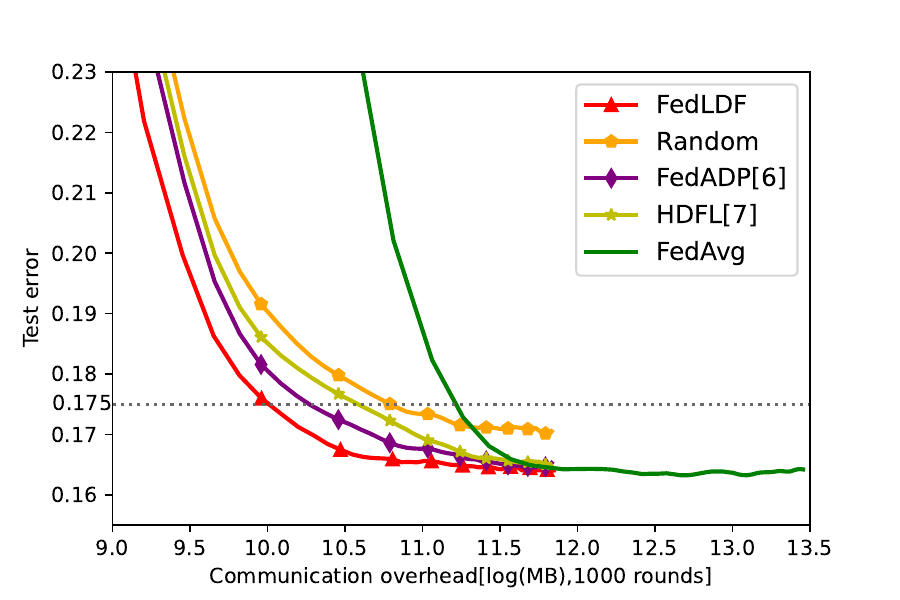}
\caption{Test error comparisons base on CIFAR10, IID.}
\label{CIFAR10,iid}
\end{figure}
\begin{figure}[!t]
\centering
\includegraphics[width=0.45\textwidth]{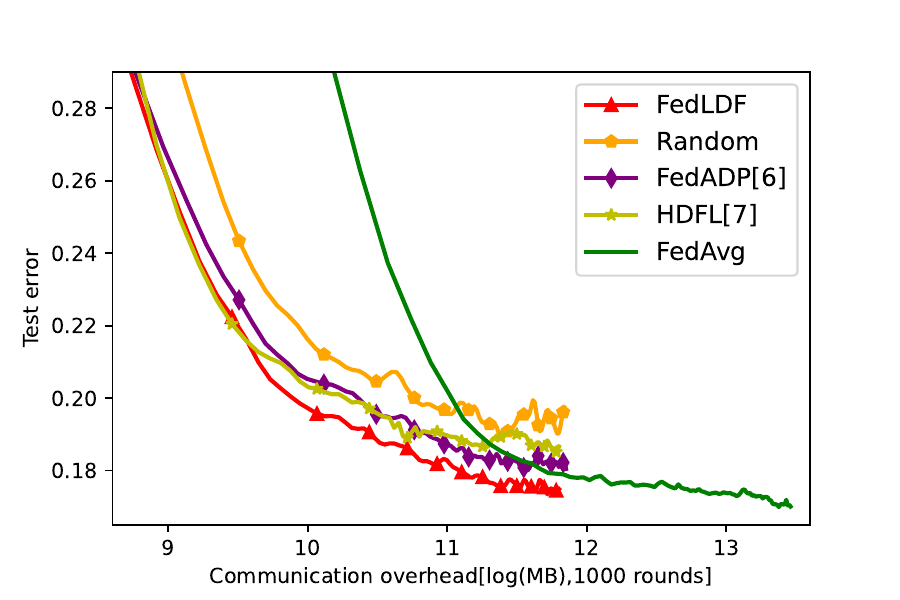}
\caption{Test error comparisons base on CIFAR10, Non-IID.}
\label{CIFAR10,Non-iid}
\end{figure}

\subsection{Performance Analysis}
In Fig. \ref{CIFAR10,iid}, the performance of listed algorithms with communication overhead is demonstrated in the i.i.d. case. For convenience to show the curve in the figure, we take the logarithm of the communication overhead. We can see that, although not all clients participate in the aggregation process, FedLDF performs even $0.4\%$ better than FedAvg at the last communication round and the communication efficiency increases $80\%$ for model uploading. This improvement is attributed to the fact that we select the clients for each layer instead of every client uploading their all layers. Similarly, from all curves, it can be observed that our algorithm achieves a faster convergence speed. When error equals $0.175$, compared with FedADP\cite{Liu2023globalcom} and HDFL\cite{zawad2023hdfl}, our algorithm saves $19.9\%$ and $35.8\%$ communication costs respectively. Moreover, there is also better model performance, reaching $1.9\%$ and $2.5\%$ at the maximum gap, respectively. This highlights the effectiveness of treating the entire layer as the smallest pruning unit of FedLDF within the FL framework. Furthermore, we observe that the performance gap between FedLDF and FedAvg diminishes as training progresses, consistent with the findings of \textbf{Theorem \ref{theorem1}}. When compared to the random algorithm, Fig. \ref{CIFAR10,iid} shows that the test error of our algorithm dropped $3.2\%$, which embodies the advantage of our algorithm in selecting clients according to the divergence. 

In Fig. \ref{CIFAR10,Non-iid}, we present the model performance in the Non-i.i.d. case.We can see that, at the same error level, our algorithm still gets a lower communication overhead. Compared to Fig. \ref{CIFAR10,iid}, we observe that FedLDF exhibits a more stable and significant model performance $(2.2\%)$ to random algorithm in Fig. \ref{CIFAR10,Non-iid}. This is because each client does not have i.i.d. data, making it crucial to select more contributory layers for uploading. Although baselines\cite{Liu2023globalcom}\cite{zawad2023hdfl} perform similarly with FedLDF at low communication overhead period, the error decline of FedLDF is more obvious later. The underlying reason is that, due to the difference in the data distribution of each client, the layer parameters of each client changed greatly in the early stage of training. Hence, the gap between our algorithm and baselines is small early, but the advantages of FedLDF are reflected in the later stage. Additionally, compared to FedADP\cite{Liu2023globalcom} and HDFL\cite{zawad2023hdfl}, our algorithm has $1.2\%$ and $0.9\%$ higher model performances, respectively. Notably, FedLDF ends up with an error rate only $0.5\%$ higher than FedAvg but with $80\%$ communication overhead savings.

\section{Conclusion}
In this work, we have presented a novel model aggregation algorithm(FedLDF) for FL. By analyzing the divergence of each layer between local and global models, our proposed mechanism enhances the FL framework with model layer divergence feedback. 
Moreover, we have employed a novel and simplified convergence analysis method to assess the expected convergence rate of the FedLDF algorithm.The extensive experiments conducted in this work have demonstrated the efficacy of our method. Specifically, our approach achieves the upload of local models with significantly lower communication overhead while preserving a high global model performance. These results underscore the potential of FedLDF in addressing the challenges associated with communication resources and model performance in FL scenarios.







\bibliographystyle{IEEEtran}
\bibliography{main}
%




\end{document}